\documentclass[11pt]{article}
\usepackage{fullpage}
\usepackage{authblk}
\usepackage{graphicx}
\usepackage{amsmath}
\usepackage{amsfonts}
\usepackage{amssymb}
\usepackage{amsthm}
\usepackage{mathrsfs}
\usepackage{bm}
\usepackage{mathtools}
\usepackage{enumitem}

\usepackage{thmtools}
\usepackage{thm-restate}

\usepackage{algorithm}
\usepackage{lipsum}

\usepackage{url}
\usepackage{hyperref}

\usepackage{caption}
\usepackage{subcaption}
\newtheorem{theorem}{Theorem}%[section]
\newtheorem{corollary}[theorem]{Corollary}

\newtheorem{lemma}{Lemma}
\newtheorem{claim}{Claim}
\newtheorem{problem}{Problem}
\newtheorem{proposition}{Proposition}
\newtheorem{remark}{Remark}
\newtheorem{assumption}{Assumption}
\newtheorem{conjecture}{Conjecture}
\newtheorem{hypothesis}{Hypothesis}
\newtheorem{observation}{Observation}

\newtheorem*{definition*}{Definition}

\makeatletter

\usepackage{xparse}
\NewDocumentCommand{\lplabel}{o m}{%
  \makebox[0pt][r]{#2\hspace*{4em}}%
  \IfNoValueF{#1}
    {\def\@currentlabel{#2}\ltx@label{#1}}
}

\renewcommand\section{%
  \@startsection{section}{1}
                {\z@}%
                {-3.5ex \@plus -1ex \@minus -.2ex}%
                {2.3ex \@plus.2ex}%
                {\large\bfseries}% 11pt
}
\renewcommand\subsection{%
  \@startsection{subsection}{2}
                {\z@}%
                {-3.25ex\@plus -1ex \@minus -.2ex}%
                {1sp}% No space after subsections
                {\normalsize\bfseries}% normal size, boldface
}
\renewcommand\subsubsection{%
  \@startsection{subsubsection}{3}
                {\z@}%
                {-3.25ex\@plus -1ex \@minus -.2ex}%
                {1sp}% No space after subsubsections
                {\normalfont\normalsize}% normal size, medium
}
\makeatother

\title{{\LARGE\bf  Population Monotonic Allocation Schemes \\for Vertex Cover Games}\thanks{This work was supported in part by National Natural Science Foundation of China (11871442, 11971447) and Fundamental Research Funds for the Central Universities (201713051, 201964006).}}

\author[1]{Han Xiao\thanks{Corresponding author. Email: hxiao@ouc.edu.cn.}}
\author[1]{Qizhi Fang}
\author[2]{Ding-Zhu Du}

\affil[1]{\small School of Mathematical Sciences, Ocean University of China, Qingdao, China}
\affil[2]{\small Department of Computer Science, University of Texas at Dallas, Richardson, TX 75080, USA}
\date{}
\begin{document}

%\numberwithin{equation}{section}

\maketitle

\openup 1.2\jot

%\hfill

\begin{abstract}
For vertex cover games (introduced by Deng et al., Math. Oper. Res., 24:751-766, 1999 \cite{DengIbar99}),
we investigate population monotonic allocation schemes (introduced by Sprumont, Games Econ. Behav., 2: 378-394, 1990 \cite{Spru90}).
We show that the existence of a population monotonic allocation scheme (PMAS for short) for vertex cover games can be determined efficiently and that a PMAS, if exists, can be constructed accordingly.
We also show that integral PMAS-es for vertex cover games can be characterized with stable matchings and be enumerated by employing Gale-Shapley algorithm (introduced by Gale and Shapley, Amer. Math. Monthly, 69:9-15, 1962 \cite{GaleShap62}).

\hfill


\end{abstract}

%\newpage
\section{Introduction}
Cooperative game theory lays out a theoretical framework for analyzing cooperation among independent participants.
An essential issue in a cooperative game is to find an adequate allocation to distribute the expected cost of the coalition to individual participants.
There are many criteria for evaluating how ``good'' an allocation is, such as fairness, stability, and so on.
Emphases on different criteria lead to different solution concepts, e.g., the core, the Shapley value, the nucleolus, the bargaining set, and the von Neumann-Morgenstern solution.
Among those solution concepts, the core which addresses the issue of stability is one of the most attractive solution concepts.

The core in a cooperative game is the set of allocations for the grand coalition (i.e., the coalition of all participants),
under which no participant can derive a better payoff by leaving the grand coalition,
either individually or as a subgroup.
However, an allocation that lies in the core does not necessarily guarantee the unhindered formation of a coalition,
as the cost allocated to participants in the current coalition may increase when a new participant joins in. 

To study allocations in an expanding coalition,
population monotonic allocation schemes (also known as cross-monotonic cost-sharing schemes) were introduced,
under which no participant of any coalition derives a worse payoff after a new participant joins in.
A PMAS gives no incentive to any participant to block the expansion of coalition and hence the grand coalition is always achieved.
Besides, PMAS-es shift the attention from allocations only for the grand coalition to allocation schemes,
which deal with partial cooperation and provide allocations for any coalition.
Moreover, the set of allocations for the grand coalition that can be reached through a PMAS can be seen as a refinement of the core:
the core provides allocations in a sense of static stability, while PMAS-es provide allocations in a sense of dynamic stability.

Populations monotonic allocation schemes were first studied by Sprumont \cite{Spru90}, where some characterizations on PMAS-es were provided.
In particular, Sprumont proved that submodularity is sufficient for existence of a PMAS.
Grahn and Voorneveld \cite{GrahVoor02} showed that every bankruptcy game admits a PMAS by indicating bankruptcy rules that give rise to a PMAS.
Norde et al. \cite{NordMore04} presented a combinatorial algorithm for computing PMAS-es in minimum cost spanning tree games.
Hamers et al. \cite{HameMiqu14} characterized the class of coloring games admitting a PMAS and provided an algorithm enumerating all integral PMAS-es.
Motivated by the work of Hamers et al. \cite{HameMiqu14}, we investigate PMAS-es for vertex cover games by generalizing the characterization of submodular vertex cover games \cite{Okam03}.
Our results are in the same spirit as the work of Hamers et al. \cite{HameMiqu14}, and are also inspired by the work of Chen et al. \cite{ChenGao19}.
We provide an efficient characterization for the class of vertex cover games admitting a PMAS and show that integral PMAS-es can be characterized with stable matchings and be enumerated by employing Gale-Shapley algorithm.

Vertex cover games studied in this paper fall into the scope of combinatorial optimization games \cite{DengIbar99,DengIbar00}, which arise from cost allocations in the minimum vertex cover problem.
Vertex cover games were first studied by Deng et al. \cite{DengIbar99},
where the algorithmic aspect of the core was investigated and a complete characterization for the balancedness of vertex cover games was presented.
In a following work \cite{DengIbar00}, Deng et al. gave a necessary and sufficient condition for the total balancedness of vertex cover games.
As opposed to the model of Deng et al. \cite{DengIbar99, DengIbar00} where players are edges,
Gusev \cite{Guse19} introduced a different class of vertex cover games where players are vertices, and investigated the application to transport networks.
In this paper, we stick to the vertex cover game introduced by Deng et al. \cite{DengIbar99} and investigate PMAS-es.

The rest of this paper is organized as follows.
Section \ref{sec:prel} is a preliminary section introducing the relevant concepts of game theory and graph theory.
In Section \ref{sec:structure}, an efficient characterization for the class of vertex cover games admitting a PMAS is presented.
Section \ref{sec:dualdescription} offers a dual-based description of PMAS-es for vertex cover games.
In Section \ref{sec:integralPMAS}, we characterize and enumerate integral PMAS-es for vertex cover games with stable matchings.
Section \ref{sec:ending} concludes the results in this paper and addresses some complexity issues in computing PMAS-es.

\section{Preliminaries}
\label{sec:prel}
This section first reviews some concepts from game theory and graph theory,
and then introduces the definition and some known results for vertex cover games.

\subsection{Cooperative game theory}
A \emph{cooperative game} is a tuple $\Gamma=(N,\gamma)$, where $N$ is the set of players and $\gamma:2^N\rightarrow \mathbb{R}$ is the characteristic function with the convention $\gamma(\emptyset)=0$.
Any subset $S$ of $N$ is called a \emph{coalition}, where $N$ is called the \emph{grand coalition}.
For coalition $S$, $\gamma(S)$ represents the total cost charged to $S$.
A cooperative game $\Gamma=(N,\gamma)$ is said \emph{monotonic} if $\gamma(S)\leq \gamma(T)$ for any $S,T\in 2^N$ with $S\subseteq T$.
A cooperative game $\Gamma=(N,\gamma)$ is said \emph{submodular} if the characteristic function $\gamma$ is submodular, i.e., $\gamma (S)+\gamma (T)\geq \gamma (S\cup T)+\gamma (S\cap T)$ for any $S,T\in 2^N$.
The \emph{subgame} of $\Gamma$ corresponding to coalition $T$, denoted by $\Gamma_T$, is a game $(T,\gamma_T)$ with $\gamma_T (S)=\gamma (S)$ for any $S\subseteq T$.

A \emph{cost allocation} of $\Gamma=(N,\gamma)$ is a vector $\boldsymbol{a}=(a_i)_{i\in N}$, which consists of proposed costs to be paid by players in the grand coalition.
A cost allocation $\boldsymbol{a}$ is said \emph{efficient} if $\sum_{i\in N}a_i=\gamma(N)$,
and said \emph{group rational} if $\sum_{i\in S}a_i\leq \gamma(S)$ for any $S\subseteq N$.
In particular, $\boldsymbol{a}$ is said \emph{individual rational} if $a_i\leq \gamma(\{i\})$ for any $i\in N$.
An \emph{imputation} of $\Gamma$ is a cost allocation that is efficient and individual rational.
The \emph{core} of $\Gamma$, denoted by $\mathcal{C}(\Gamma)$, is the set of imputations that are group rational.
\iffalse
\begin{equation*}
  \mathcal{C}(\Gamma)=\Bigg\{\boldsymbol{a}~ \Bigg|~ \sum_{i\in N}a_i=\gamma(N)\text{~and~} \sum_{i\in S}a_i\leq \gamma(S)\text{~for~any~} S\subseteq N\Bigg\}.
\end{equation*}
\fi
A \emph{core allocation} is a cost allocation in the core, which satisfies all players in the grand coalition and no player has an incentive to split off from the grand coalition.
A game $\Gamma$ is said \emph{balanced} if $\mathcal{C}(\Gamma)\not=\emptyset$ and \emph{total balanced} if $\mathcal{C}(\Gamma_T)\not=\emptyset$ for any nonempty $T\subseteq N$.

A \emph{population monotonic allocation scheme} (\emph{PMAS} for short) of $\Gamma=(N,\gamma)$ is a vector $\boldsymbol{a}=(a_{S,i})_{S\in 2^N\backslash\{\emptyset\}, i\in S}$ satisfying the following two conditions:
\begin{itemize}
	\item[\textendash] \emph{efficiency}: $\sum_{i\in S}a_{S,i}=\gamma(S)$ for any $S\in 2^N\backslash \{\emptyset\}$;
	\item[\textendash] \emph{monotonicity}: $a_{S,i}\geq a_{T,i}$ for any $S, T\in 2^N\backslash \{\emptyset\}$ with $S\subseteq T$ and any $i\in S$.
\end{itemize}
Let $\mathcal{P}(\Gamma)$ be the set of PMAS-es for $\Gamma$.
\iffalse
\begin{equation*}
  \mathcal{P}(\Gamma)=
  \Bigg\{\boldsymbol{a}~\Bigg|~
  \begin{aligned}
  &\mbox{$\sum_{i\in S} a_{S,i}$}=\gamma(S) \text{~for~any~} S\in 2^N\backslash \{\emptyset\} ;\\
  &a_{S,i}\geq a_{T,i} \text{~for~any~} S,T\in 2^N\backslash \{\emptyset\} \text{~with~} S\subseteq T \text{~and~any~} i\in S
  \end{aligned}
  \Bigg\}.
\end{equation*}
\fi
When $ \mathcal{P}(\Gamma)\not=\emptyset$, $\Gamma$ is said \emph{population monotonic} (also known as \emph{cross-monotonic}).
Let $\boldsymbol{a}=(a_{S,i})_{S\in 2^N\backslash \{\emptyset\}, i\in S}$ be a PMAS for $\Gamma$.
For any $S\in 2^N\backslash \{\emptyset\}$, denote by $\boldsymbol{a}_S=(a_{S,i})_{i\in S}$ the restriction of $\boldsymbol{a}$ to $S$, which is a core allocation of subgame $\Gamma_S$.
Notice that even total balancedness is not sufficient for the population monotonicity,
as a PMAS provides for every coalition a core allocation in a cross-monotonic way.
Hamers et al. \cite{HameMiqu14} proved that PMAS-es of monotonic cooperative games are always nonnegative.
We refer to \cite{Spru90} for more about PMAS-es.
\begin{lemma}[Hamers et al. \cite{HameMiqu14}]
  \label{thm:nonnegative}
  Let $\boldsymbol{a}$ be a PMAS for a monotonic cooperative game $\Gamma=(N,\gamma)$.
  Then $a_{S,i}\geq 0$ for any $S\in 2^N\backslash \{\emptyset\}$ and any $i\in S$.
\end{lemma}

\subsection{Graph theory}
Throughout, a graph is always finite, undirected and simple.
Let $n\in \mathbb{N}$.
We use $K_n$ to denote the complete graph with $n$ vertices,
use $C_n$ to denote the graph which is a cycle with $n$ vertices,
and use $P_n$ to denote the graph which is a path with $n$ vertices.
Let $H$ be a graph.
We use $V(H)$ to denote the vertex set of $H$ and use $E(H)$ to denote the edge set of $H$.
A graph is said $H$-\emph{free} if it contains no subgraph isomorphic to $H$.
A graph is \emph{bipartite} if it is odd-cycle-free.
A graph is a \emph{forest} if it is cycle-free.
A forest is a \emph{tree} if it is connected.
A vertex is \emph{pendant} if it has degree one.
An edge is \emph{pendant} if it is incident to a pendant vertex.

Let $G=(V,E)$ be a graph.
The \emph{distance} of two vertices $u$ and $v$ in $G$ is the minimum number of edges in a path connecting them.
The \emph{diameter} of $G$ is the largest distance between any two vertices in $G$.
For any $S\subseteq E$,
$G[S]$ denotes the edge-induced subgraph of $G$,
$V_S$ denotes the vertex set of $G[S]$,
and $\delta_S(v)$ denotes the set of edges incident to $v$ in $G[S]$.
A \emph{vertex cover} of $G$ is a vertex set $C\subseteq V$ such that each edge of $G$ intersects $C$.
The \emph{vertex cover number} of $G$, denoted by $\tau(G)$, is the minimum size of vertex covers in $G$.
A \emph{matching} of $G$ is an edge set $M\subseteq E$ without common vertices.
The \emph{matching number} of $G$, denoted by $\nu(G)$, is the maximum size of matchings in $G$.
Clearly, $\nu(G)\leq \tau (G)$,
since every vertex in a vertex cover only covers at most one edge in a matching.
It is well known that equality $\nu(G)= \tau (G)$ holds when $G$ is bipartite \cite{Schr03}.

\subsection{Vertex cover games}
A vertex cover game has players on edges and the game value is defined by the vertex cover number.
Formally, the \emph{vertex cover game} on a graph $G=(V,E)$ is a cooperative game $\Gamma_G=(N,\gamma)$, where $N=E$ and $\gamma (S)=\tau(G[S])$ for any $S\subseteq N$.

\begin{lemma}
  \label{lemma:monotonicity}
  Every vertex cover game is monotonic.
\end{lemma}
\begin{proof}
  Let $\Gamma_G=(N,\gamma)$ be the vertex cover game on a graph $G$.
  Let $S\subseteq T\subseteq N$.
  Notice that every vertex cover of $G[T]$ is also a vertex cover of $G[S]$.
  It follows that $\gamma(S)\leq \gamma(T)$.
\end{proof}

\begin{lemma}[Deng et al. \cite{DengIbar99}]
The vertex cover game $\Gamma_G$ on a graph $G$ is balanced if and only if $\nu(G)=\tau(G)$.
\end{lemma}

\begin{lemma}[Deng et al. \cite{DengIbar00}]
The vertex cover game $\Gamma_G$ on a graph $G$ is totally balanced if and only if $G$ is bipartite.
\end{lemma}

\begin{lemma}[Okamoto \cite{Okam03}]
The vertex cover game $\Gamma_G$ on a graph $G$ is submodular if and only if $G$ is $(K_3,P_4)$-free.
\end{lemma}

\section{An efficient characterization for population monotonicity}
\label{sec:structure}
In this section, we show that a graph induces a vertex cover game admitting PMAS-es if and only if the graph is $(K_3,C_4,P_5)$-free.
We decompose our proof into several lemmas.

\begin{figure}[h]
\vspace{-1.5em}
\centering
\includegraphics[width=0.71\textwidth]{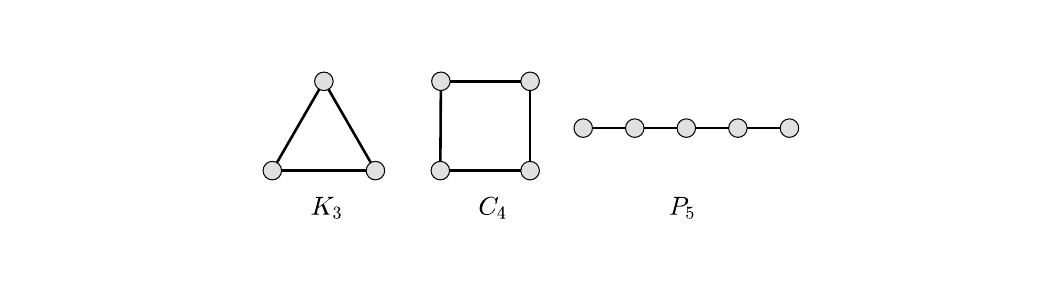}
\vspace{-2em}
\caption{Forbidden subgraphs}
\end{figure}
\vspace{-1em}

\begin{lemma}
\label{thm:onlyif}
Let $\Gamma_G=(N,\gamma)$ be the vertex cover game on a graph $G$.
If $\mathcal{P}(\Gamma_G)\not=\emptyset$, then $G$ is $(K_3,C_4,P_5)$-free.
\end{lemma}

\begin{proof}
Let $\boldsymbol{a}\in \mathcal{P}(\Gamma_G)$.
We show that any of $K_3$, $C_4$ and $P_5$ yields a contradiction.

We first consider $K_3$.
Suppose $E(K_3)=\{1,2,3\}$.
Note that $\gamma(\{1,2,3\})=2$ and $\gamma(\{1,2\})=\gamma(\{2,3\})=\gamma(\{1,3\}=1$.
By efficiency and monotonicity, we have
\begin{equation*}
\begin{split}
 4
 & = \gamma(\{1,2,3\})+\gamma(\{1,2,3\}) \\
 & = a_{\{1,2,3\},1}+ a_{\{1,2,3\},2}+ a_{\{1,2,3\},3}+ a_{\{1,2,3\},1}+ a_{\{1,2,3\},2}+ a_{\{1,2,3\},3}\\
 & \leq a_{\{1,2\},1}+ a_{\{1,2\},2}+ a_{\{2,3\},3}+ a_{\{1,3\},1}+ a_{\{2,3\},2}+ a_{\{1,3\},3}\\
 & = \gamma(\{1,2\})+\gamma(\{2,3\})+\gamma(\{1,3\})\\
 & =3,
\end{split}
\end{equation*}
which yields a contradiction.

Now we check $C_4$ and $P_5$.
Let $H$ be either $C_4$ or $P_5$.
Suppose $E(H)=\{1,2,3,4\}$, where $i$ and $i+1$ are incident in $H$.
Note that $\gamma(\{1,2,3\})=\gamma(\{2,3,4\})=2$ and $\gamma(\{1,2\})=\gamma(\{2,3\})=\gamma(\{3,4\})=1$.
By efficiency and monotonicity, we have
\begin{equation*}
\begin{split}
 4
 & = \gamma(\{1,2,3\})+\gamma(\{2,3,4\}) \\
 & = a_{\{1,2,3\},1}+ a_{\{1,2,3\},2}+ a_{\{1,2,3\},3}+ a_{\{2,3,4\},2}+ a_{\{2,3,4\},3}+ a_{\{2,3,4\},4}\\
 & \leq a_{\{1,2\},1}+ a_{\{1,2\},2}+ a_{\{2,3\},3}+ a_{\{2,3\},2}+ a_{\{3,4\},3}+ a_{\{3,4\},4}\\
 & = \gamma(\{1,2\})+\gamma(\{2,3\})+\gamma(\{3,4\})\\
 & =3,
\end{split}
\end{equation*}
which yields a contradiction.
\end{proof}

The following lemma gives an alternative characterization for $(K_3,C_4,P_5)$-free graphs.
\begin{lemma}
\label{lemma:structure}
A graph $G$ is $(K_3,C_4,P_5)$-free if and only if each component of $G$ is a tree of diameter at most $3$.
\end{lemma}
\begin{proof}
We first prove the ``only if'' part.
Notice that a $(K_3, P_5)$-free graph does not contain any odd cycle and that a $(C_4,P_5)$-free graph does not contain any even cycle.
It follows that every $(K_3,C_4,P_5)$-free graph is a forest.
Since any tree of diameter larger than $3$ contains a $P_5$,
each component of a $(K_3,C_4,P_5)$-free graph is a tree of diameter at most $3$.

Now we prove the ``if'' part.
Let $G$ be a graph whose components are trees of diameter at most $3$.
Thus $G$ is $(K_3,C_4)$-free.
Since any graph of diameter at most $3$ is $P_5$-free, $G$ is $(K_3,C_4,P_5)$-free.
\end{proof}

\begin{figure}[h]
\vspace{-3em}
    \centering
    \begin{minipage}{0.45\textwidth}
        \centering
        \includegraphics[width=0.45\textwidth]{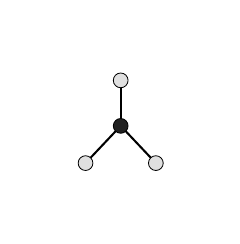}
        \vspace{-2.5em}
        \caption{An example of stars. The dark vertex is the center.}
        \label{fig:star}
    \end{minipage}\hfill
    \begin{minipage}{0.515\textwidth}
        \centering
        \includegraphics[width=.85\textwidth]{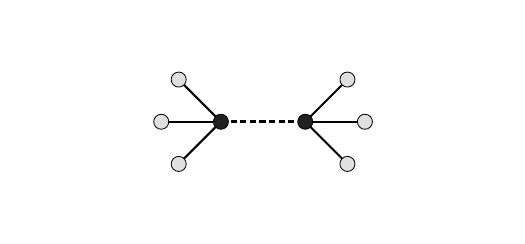}
        \vspace{-2.5em}
        \caption{An example of pisceses. The dark vertices are the bases and the dashed edge is a free rider.}
         \label{fig:pisces}
    \end{minipage}
\end{figure}
\vspace{-.5em}

Before proceeding, we introduce some notions for simplicity.
A tree of diameter $2$ is called a \emph{star}.
The unique non-pendant vertex of a star is called the \emph{center} (see Figure \ref{fig:star}).
Clearly, the center of a star is a minimum vertex cover for the star.
A $K_2$ can also be viewed as a star but only one endpoint can be viewed as the center.
A tree of diameter $3$ is called a \emph{pisces}.
A pisces can be obtained from two stars by joining their centers with an edge.
The two non-pendant vertices in a pisces are called the \emph{bases} which form a minimum vertex cover for the pisces.
The unique non-pendant edge in a pisces is called a \emph{free rider} which has special significance for vertex cover games (see Figure \ref{fig:pisces}).
To see this, consider the vertex cover game on a pisces.
For any coalition without the non-pendant edge, there is a minimum vertex cover which is a subset of the two bases.
Hence the non-pendent edge in a pisces can alway take a free ride and get covered by a minimum vertex cover of other edges.

Let $G=(V,E)$ be $(K_3,C_4,P_5)$-free.
Lemma \ref{lemma:structure} implies that every component of $G$ is either a star or a pisces.
Let $C^*\subseteq V$ be the set of centers of stars and bases of pisceses in $G$.
Then $C^*$ is a minimum vertex cover of $G$.
For any nonempty $S\subseteq E$, there is a minimum vertex cover $C^*_S\subseteq C^*$ of $G[S]$, as every component of $G[S]$ is either a star or a pisces.
We will use these notations repeatedly in the rest of this paper.
Now we are ready to present one of our main results.

\begin{theorem}
\label{thm:PM}
  Let $\Gamma_G=(N,\gamma)$ be the vertex cover game on a graph $G$.
  Then $\Gamma_G$ is population monotonic if and only if $G$ is $(K_3,C_4,P_5)$-free.
\end{theorem}

\begin{proof}
The ``only if'' part follows from Lemma \ref{thm:onlyif}.
Now we prove the ``if'' part.
Assume that $G$ is $(K_3,C_4,P_5)$-free and $\Gamma_G$ is the vertex cover game on $G$.
By Lemma \ref{lemma:structure}, every component of $G$ is either a star or a pisces.
Let $C^*$ be the set of centers of stars and bases of pisceses in $G$.
For any $S\in 2^N\backslash \{\emptyset\}$,
let $C^*_S\subseteq C^*$ be a minimum vertex cover of $G[S]$,
and define the cost allocation by
\begin{equation*}
a_{S,i}=
\begin{cases}
\:\:\,\, 0 & \text{if $i$ is a free rider incident to other edges in $S$,}
\vspace{1.5 mm}\\
\:\:\,\, 1 & \text{if $i$ is a free rider not incident to other edges in $S$,}
\vspace{1.5 mm}\\
\frac{1}{\lambda_S(i)}& \text{otherwise,}
\end{cases}
\end{equation*}
where $\lambda_S(i)$ is the number of non-free rider edges in $S$ incident to the vertex in $C^*_S$ that covers $i$.
The idea behind this allocation scheme is simple.
For any $S\in 2^N\backslash \{\emptyset\}$,
split every vertex in $C^*_S$ equally among non-free rider edges in $S$ covered by the vertex,
unless a free rider is the unique edge in $S$ covered by it.
It remains to show that $(a_{S,i})_{S\in 2^N\backslash \{\emptyset\},i\in S}$ is a PMAS for $\Gamma_G$.

We first prove efficiency.
Let $S\in 2^N\backslash \{\emptyset\}$.
By the construction, we have $\sum_{i\in \delta_S(v)} a_{S,i}=1$ for any $v\in C^*_S$.
Further notice that free riders are the only possible edges incident to more than one vertex in $C^*_S$.
Since $a_{S,i}=0$ for any free rider $i$ that is incident to other edges in $S$, we have
\begin{equation*}
\sum_{i\in S} a_{S,i}=\sum_{v\in C^*_S}\sum_{i\in \delta_S(v)} a_{S,i}=\lvert C^*_S\rvert=\gamma(S).
\end{equation*}

We now check monotonicity.
%$i.e., $a_{S,i}\geq a_{T,i}$ for any $S,T\in 2^N\backslash \{\emptyset\}$ with $S\subseteq T$ and $i \in S$.
Let $S,T\in 2^N\backslash \{\emptyset\}$ with $S\subseteq T$ and let $i \in S$.
We distinguish two cases of $i$.
First assume that $i$ is a free rider.
We have $1=a_{S,i}\geq a_{T,i}\geq 0$ if $i$ is not incident to other edges in $S$ and $  a_{S,i}=a_{T,i}=0$ otherwise.
It follows that $a_{S,i}\geq a_{T,i}$ when $i$ is a free rider.
Now assume that $i$ is not a free rider.
Since $S\subseteq T$, we have $\lambda_S(i)\leq \lambda_T(i)$, implying that
\begin{equation*}
  a_{S,i}=\frac{1}{\lambda_{S}(i)}\geq \frac{1}{\lambda_{T} (i)}=a_{T,i}.
\end{equation*}
Therefore, $a_{S,i}\geq a_{T,i}$ follows in either case.
\end{proof}

Our proof for Theorem \ref{thm:PM} is constructive, which provides a PMAS for every population monotonic vertex cover game and motivates our subsequent work.
The PMAS in our proof is based on a simple principle:
for every vertex in a special minimum vertex cover, share the cost equally among non-free rider players covered by the vertex.
By Lemma \ref{lemma:structure}, every component of a $(K_3$, $C_4$, $P_5)$-free graph has at most two non-pendant vertices.
Hence a $(K_3$, $C_4$, $P_5)$-free graph can be recognized efficiently, which implies that the population monotonicity of vertex cover games can be determined efficiently.

\begin{corollary}
The population monotonicity of vertex cover games can be determined in polynomial time and a PMAS, if exists, can be constructed accordingly.
\end{corollary}

\section{A dual-based description of PMAS-es with free riders}
\label{sec:dualdescription}

Let $G=(V,E)$ be a graph and $S\subseteq E$ be a nonempty edge set.
Denote by $LP(S)$ the following linear program defined on $G[S]$.
\begin{alignat}{2}
\min\quad & \sum_{v \in V}y_{v} &{}& \nonumber \label{eqn:lp}\\
\lplabel[lp]{$LP(S)$:}\mbox{s.t.}\quad
 &y_{u}+y_{v} &\geq 1, \quad& \forall~ i_{uv} \in S, \\
 &~\quad\quad y_v&\geq 0, \quad& \forall ~v\in V_S.
\end{alignat}
The incidence vector of any minimum vertex cover in $G[S]$ is a feasible solution of $LP(S)$.
It follows that $\tau(G[S])$ is lower bounded by the optimum of $LP(S)$.
K\"{o}nig Theorem states that the gap between $\tau(G[S])$ and the optimum of $LP(S)$ is closed when $G[S]$ is bipartite.
Denote by $DP(S)$ the dual of $LP(S)$.
Hence $\tau(G[S])$ equals the optimum of $DP(S)$ when $G[S]$ is bipartite.
\begin{alignat}{2}
\max\quad & \sum_{i \in S}x_{S,i} &{}& \nonumber\\
\lplabel[lp:dual]{$DP(S)$:}\mbox{s.t.}\quad
& \sum_{i\in \delta_S(v)} x_{S,i}\leq 1, &\quad& \forall~ v \in V_S, \label{ineq:vertex}\\
&~~~\quad\quad x_{S,i} \geq 0, &{}& \forall~ i \in S. \label{ineq:edge}
\end{alignat}

Further assume that $G$ is $(K_3,C_4,P_5)$-free.
Let $\Gamma_G=(N,\gamma)$ be the vertex cover game on $G$.
For any $S\in 2^N\backslash \{\emptyset\}$, $\gamma(S)$ equals the optimum of $DP(S)$.
Moreover, we have the following observation for PMAS-es of vertex cover games.

\begin{lemma}
\label{thm:dualallocation}
Let $\Gamma_G=(N,\gamma)$ be the vertex cover game on a graph $G$ and $\boldsymbol{a}$ be a PMAS for $\Gamma_G$.
Then $\boldsymbol{a}_S$ is an optimal solution of $DP(S)$ for any $S\in 2^N\backslash \{\emptyset\}$.
\end{lemma}
\begin{proof}
We first prove that $\boldsymbol{a}_S$ is feasible to $DP(S)$.
Since vertex cover games are monotonic, the nonnegativity of $\boldsymbol{a}_S$ follows from Lemma \ref{thm:nonnegative}.
It remains to show that $\sum_{i\in \delta_S(v)} a_{S,i}\leq 1$  for any $v\in V_S$.
Let $v\in V_S$.
Clearly, $\delta_S(v)\subseteq S$.
By efficiency and monotonicity, we have
\begin{equation*}
  \sum_{i\in \delta_S(v)} a_{S,i}\leq \sum_{i\in \delta_S(v)}a_{\delta_S(v),i}=\gamma(\delta_S(v))=1.
\end{equation*}
Hence $\boldsymbol{a}_S$ is feasible to $DP(S)$.

Now we prove that $\boldsymbol{a}_S$ is optimal to $DP(S)$.
By efficiency, we have
\begin{equation*}
  \sum_{i\in S} a_{S,i}=\gamma(S).
\end{equation*}
Theorem \ref{thm:PM} implies that $G[S]$ is $(K_3,C_4,P_5)$-free, which is a special bipartite graph.
Hence $\gamma(S)$ equals the optimum of $DP(S)$, implying that $\boldsymbol{a}_S$ is an optimal solution of $DP(S)$.
\end{proof}

Lemma \ref{thm:dualallocation} implies that for vertex cover games, allocations for each coalition in a PMAS are dual optimal solutions of corresponding minimum vertex cover problem.
In the following, we present a more precise dual-based description for PMAS-es with free riders.
Let $C^*$ be the set of centers of stars and bases of pisceses in $G$ and $C^*_S\subseteq C^*$ be a minimum vertex cover of $G[S]$.
Let $\pi(S)$ be the linear system \eqref{ineq:vertex}-\eqref{ineq:edge} and $\pi^* (S)$ be the linear system obtained from $\pi (S)$ by
\begin{itemize}
	\item[\textendash] setting \eqref{ineq:vertex} to $\sum_{i\in \delta_S(v)} x_{S,i}= 1$ for any $v\in C^*_S$, and
	\item[\textendash] setting \eqref{ineq:edge} to $x_{S,i}=0$ for any free rider $i$ incident to other edges in $S$.
\end{itemize}
Clearly, $\pi^* (S)\subseteq \pi(S)$.
Moreover, $\pi^* (S)$ is a subset of optimal solutions of $DP(S)$.

\begin{lemma}
Let $G=(V,E)$ be a $(K_3,C_4,P_5)$-free graph.
Then every vector in $\pi^* (S)$ is an optimal solutions of $DP(S)$ for any nonempty $S\subseteq E$.
\end{lemma}
\begin{proof}
Let $C^*$ be the set of centers of stars and bases of pisceses in $G$ and $S\subseteq E$ be a nonempty edge set.
Let $C^*_S\subseteq C^*$ be a minimum vertex cover of $G[S]$ and $\boldsymbol{x}_S\in \pi^* (S)$.
Since $\pi^* (S)\subseteq \pi(S)$, we have $\boldsymbol{x}_S\in \pi(S)$.
It remains to prove the optimality of $\boldsymbol{x}_S$.
Lemma \ref{lemma:structure} implies that every component of $G[S]$ is either a star or a pisces.
Hence free riders in $S$ are the only possible edges incident to more than one vertex in $C^*_S$.
Since $x_{S,i}=0$ for any free rider $i$ incident to other edges in $S$,
it follows that
\begin{equation*}
\sum_{i\in S} x_{S,i}=\sum_{v\in C^*_S}\sum_{i\in \delta_S (v)}x_{S,i}=\lvert C^*_S\rvert=\gamma(S).
\end{equation*}
Therefore, $\boldsymbol{x}_S$ is an optimal solution of $DP(S)$.
\end{proof}

Now we strengthen Lemma \ref{thm:dualallocation} and present a dual-based description of PMAS-es with free riders.
\begin{theorem}
\label{thm:DualDescription_FreeRider}
Let $\Gamma_G=(N,\gamma)$ be the vertex cover game on a graph $G$ and $\boldsymbol{a}$ be a PMAS for $\Gamma_G$.
Then $\boldsymbol{a}_S \in \pi^* (S)$ for any $S\in 2^N\backslash \{\emptyset\}$.
\end{theorem}

\begin{proof}
Let $\boldsymbol{a}$ be a PMAS for $\Gamma_G$.
Theorem \ref{thm:PM} and Lemma \ref{lemma:structure} imply that every component of $G$ is either a star or a pisces.
Let $C^*$ be the set of centers of stars and bases of pisceses in $G$.
Let $S\in 2^N\backslash \{\emptyset\}$ and $C^*_S\subseteq C^*$ be a minimum vertex cover of $G[S]$.
Lemma \ref{thm:dualallocation} implies that $\boldsymbol{a}_S$ is an optimal solution of $DP(S)$.
Notice that any minimum vertex cover for a graph is a union of minimum vertex covers for every component.
Hence we assume that $G[S]$ is connected.
We show that $\boldsymbol{a}_S$ satisfies all equality constraints in $\pi^* (S)$ by distinguishing three cases.

Case $1$: $G[S]$ is a star without free riders.
Let $v^*$ be the center of $G[S]$.
Notice that $\delta_S (v^*)=S$ and $C^*_S=\{v^*\}$.
By efficiency, we have
\begin{equation*}
  \sum_{i\in \delta_S (v^*)} a_{S,i}=\sum_{i\in S} a_{S,i}=\gamma(S)=1.
\end{equation*}

Case $2$: $G[S]$ is a pisces.
Let $i^*$ be the free rider in $S$ and $v^*_1, v^*_2$ be the endpoints of $i^*$.
Hence $C^*_S=\{v^*_1,v^*_2\}$, $\delta_S(v^*_1)\cup \delta_S(v^*_2)=S$ and $\delta_S(v^*_1)\cap \delta_S(v^*_2)=\{i^*\}$.
By efficiency and monotonicity, we have  
\begin{equation*}
   \begin{split}
      2 = \gamma(S)
      & =\sum_{i\in S} a_{S,i}\\
      & \leq \sum_{i\in \delta_S(v^*_1)} a_{\delta_S(v^*_1),i}+\sum_{i\in \delta_S(v^*_2)} a_{\delta_S(v^*_2),i}-a_{\delta_S(v^*_k),i^*}\\
     & = \gamma (\delta_S(v^*_1))+\gamma (\delta_S(v^*_2)) -a_{\delta_S(v^*_k),i^*}\\
     & = 2-a_{\delta_S(v^*_k),i^*}
   \end{split}
\end{equation*}
for $k=1,2$.
By monotonicity, we have
\begin{equation*}
a_{S,i^*}=a_{\delta_S(v^*_k),i^*}=0
\end{equation*}
and hence
\begin{equation*}
\sum_{i\in \delta_S(v^*_k)} a_{S,i}=\sum_{i\in \delta_S(v^*_k)} a_{\delta_S(v^*_k),i}=1
\end{equation*}
for $k=1,2$.

Case $3$: $G[S]$ is a star with a free rider.
Let $v^*_1$ be the center of $G[S]$.
Notice that $\delta_S (v^*_1)=S$ and $C^*_S=\{v^*_1\}$.
Since $G[S]$ contains a free rider, there exists $T\subseteq N$ such that $G[T]$ is a pisces and $\delta_S (v^*_1)=\delta_T (v^*_1)$.
Let $i^*$ be the free rider in $S\subsetneq T$ and $v^*_1, v^*_2$ be the endpoints of $i^*$.
Let $C^*_T=\{v^*_1,v^*_2\}$.
Clearly, $C^*_T\subseteq C^*$ is a minimum vertex cover of $G[T]$.
%We have seen in Case $2$ that $a_{T,i^*}=0$ and $\sum_{i\in \delta_T (v^*_1)} a_{T,i}=1$.
By monotonicity, we have
\begin{equation*}
\sum_{i\in \delta_S (v^*_1)} a_{S,i}=\sum_{i\in \delta_T (v^*_1)} a_{T,i}=1
\end{equation*}
and hence
\begin{equation*}
a_{S,i^*}=a_{T,i^*}=0.
\end{equation*}
\end{proof}

\section{Integral PMAS-es and stable matchings}
\label{sec:integralPMAS}

In this section,
we first show that every integral PMAS for a vertex cover game $\Gamma_G$ is an extreme point of $\mathcal{P}(\Gamma_G)$,
then use stable matchings to characterize integral PMAS-es for $\Gamma_G$,
and finally conclude that integral PMAS-es for $\Gamma_G$ can be enumerated by employing Gale-Shapley algorithm.

\begin{theorem}
\label{thm:integralPMAS}
  Let $\Gamma_G=(N,\gamma)$ be the vertex cover game on a graph $G$ with $\mathcal{P}(\Gamma_G)\not=\emptyset$.
  Then every integral PMAS for $\Gamma_G$ is an extreme point of $\mathcal{P}(\Gamma_G)$.
\end{theorem}
\begin{proof}
Let $\boldsymbol{a}$ be an integral PMAS for $\Gamma_G$.
Suppose $\boldsymbol{a}=\frac{1}{2} \boldsymbol{b}+\frac{1}{2} \boldsymbol{c}$, where $\boldsymbol{b}, \boldsymbol{c}\in \mathcal{P}(\Gamma_G)$.
Let $S\in 2^N\backslash \{\emptyset\}$.
It follows that $\boldsymbol{a}_S=\frac{1}{2}\boldsymbol{b}_S+\frac{1}{2}\boldsymbol{c}_S$.
By Lemma \ref{thm:dualallocation}, $\boldsymbol{a}_S$, $\boldsymbol{b}_S$ and $\boldsymbol{c}_S$ are all optimal solutions of $DP(S)$, implying that $0\leq a_{S,i}\leq 1$, $0\leq b_{S,i}\leq 1$, and $0\leq c_{S,i}\leq 1$ for any $i\in S$.
Since $\boldsymbol{a}_S$ is integral, we have $a_{S,i}=b_{S,i}=c_{S,i}$ for any $i\in S$.
Therefore, $\boldsymbol{a}$ is an extreme point of $\mathcal{P}(\Gamma_G)$.
\end{proof}

Theorem \ref{thm:integralPMAS} states that every integral PMAS for a vertex cover game $\Gamma_G$ is an extreme point of $\mathcal{P}(\Gamma_G)$.
Unfortunately, not all extreme points of $\mathcal{P}(\Gamma_G)$ are integral even when $G$ is a star.
Consider the vertex cover game $\Gamma_{K_{1,n}}=(N,\gamma)$ where $n\geq 4$.
Notice that $\gamma(S)=1$ for any $S\in 2^N\backslash \{\emptyset\}$.
Thus $\Gamma_{K_{1,n}}$ falls into the scope of \emph{unit games} investigated by Hamers et al. \cite{HameMiqu14},
where they showed that $\mathcal{P}(\Gamma_{K_{1,n}})$ has more than $(n-2)\cdot n!$ non-integral extreme points.
Hence for $\mathcal{P}(\Gamma_G)$, instead of all extreme points, we focus on integral extreme points.

Before proceeding, we introduce the notion of preference systems and stable matchings.
Let $G=(V,E)$ be a graph.
For any $v\in V$, let $\prec_v$ be a strict linear order on edges in $\delta(v)$.
We call $\prec_v$ the \emph{preference} of $v$.
For any $i,j\in \delta(v)$, we say that $i$ \emph{dominates} $j$ (at $v$) if $i\prec_v j$.
We use $\prec$ to denote the set of preferences $\prec_v$ for any $v\in V$,
and call the ordered pair $(G,\prec)$ a \emph{preference system}.
In particular, $(G,\prec)$ is \emph{bipartite} if $G$ is bipartite.
A \emph{stable matching} in $(G,\prec)$ is a matching $M$ of $G$ such that every edge in $E\backslash M$ is dominated by an edge in $M$.
For any $S\subseteq E$, we use $\prec_S$ to denote the restriction of $\prec$ to $S$,
and hence $(G[S],\prec_S)$ is also a preference system.
Gale and Shapley \cite{GaleShap62} proved that every bipartite preference system admits a stable matching by providing an efficient algorithm, namely Gale-Shapley algorithm, for computing stable matchings.
Now we are ready to characterize integral PMAS-es for vertex cover games with stable matchings.

\begin{theorem}
\label{thm:PMAS2SM}
Let $\Gamma_G=(N,\gamma)$ be the vertex cover game on a graph $G$ with $\mathcal{P}(\Gamma_G)\not=\emptyset$.
Then the following statements are equivalent.
\begin{enumerate}[label={\emph{($\roman*$)}}]
	\item $\boldsymbol{a}$ is an integral PMAS for $\Gamma_G$.
	\item There is a preference system $(G,\prec)$, where every free rider has the lowest rank, such that $\boldsymbol{a}_S$ is the incidence vector of a stable matching in $(G[S],\prec_S)$ for any  $S\in 2^N \backslash \{\emptyset\}$.
\end{enumerate}
\end{theorem}
\begin{proof}
Theorem \ref{thm:PM} and Lemma \ref{lemma:structure} imply that every component of $G$ is either a star or a pisces.

$(i)\Rightarrow (ii)$.
Let $\boldsymbol{a}$ be an integral PMAS for $\Gamma_G$.
We first define a preference system $(G,\prec)$ from $\boldsymbol{a}$.
Let $C^*$ be the set of centers of stars and bases of pisceses in $G$.
To define a preference system $(G,\prec)$, it suffices to define a preference $\prec_v$ for any $v\in C^*$.
Let $v^*$ be a vertex in $C^*$.
For any nonempty set $S\subseteq \delta(v^*)$, we have $\sum_{i\in S} a_{S,i}=\gamma(S)=1$.
We define a preference $\prec_{v^*}$ for $v^*$ from $\boldsymbol{a}$ as follows.
Start with $S=\delta(v^*)$.
  Let $i^*\in S$ be the edge with $a_{S,i^*}=1$.
  Define partial orders of $\prec_{v^*}$ by
  \begin{equation*}
    i^* \prec_{v^*} j, ~\forall\; j\in S\backslash \{i^*\}.
  \end{equation*}
  Update $S$ with $S\backslash \{i^*\}$ and repeat the process above until $S=\emptyset$.
  We claim that when $S$ becomes empty, $\prec_{v^*}$ is a well-defined strict linear order on $\delta(v^*)$.
  Indeed, consider any two edges $i,j$ in $\delta(v^*)$.
  Clearly, $G[\{i,j\}]$ is a star.
  Since $a_{\{i,j\},i}+a_{\{i,j\},j}=\gamma (\{i,j\})=1$, we may assume that $a_{\{i,j\},i}=1$ and $a_{\{i,j\},j}=0$.
  For any $S\subseteq N$ with $\{i,j\}\subseteq S$, the monotonicity implies that $a_{S,j}=a_{\{i,j\},j}=0$.
  Hence $i\prec_{v^*} j$ always holds.
  Thus $\boldsymbol{a}$ determines a unique preference system $(G,\prec)$.
  
  Now we show that every free rider in $(G,\prec)$ has the lowest rank in any preference.
  Let $i^*$ be a free rider in $G$ and $v^*_1, v^*_2$ be the endpoints of $i^*$.
  Clearly, $v^*_1,v^*_2\in C^*$.
  By Theorem \ref{thm:DualDescription_FreeRider}, $a_{S,i^*}=0$ for any $S$ containing other edges incident to $i^*$.
  Thus $i^*$ has the lowest rank in both $\prec_{v^*_1}$ and $\prec_{v^*_2}$.
  
  Therefore, $\boldsymbol{a}$ defines a unique preference system $(G,\prec)$ where every free rider has the lowest rank.
  Since every component of $G$ is either a star or a pisces, it is easy to see that $(G[S],\prec_S)$ has a unique stable matching $M_S$ with incidence vector $\boldsymbol{a}_S$ for any $S\in 2^N\backslash \{\emptyset\}$.

$(ii)\Rightarrow (i)$.
  Let $(G,\prec)$ be a preference system where every free rider has the lowest rank.
  Since every component of $G$ is either a star or a pisces, $(G[S],\prec_S)$ has a unique stable matching $M_S$ for any $S\in 2^N\backslash \{\emptyset\}$.
  We show that $\boldsymbol{a}=(a_{S,i})_{S\in 2^N\backslash \{\emptyset\},i\in S}$ is an integral PMAS for $\Gamma_G$, where $\boldsymbol{a}_S=(a_{S,i})_{i\in S}$ is the incidence vector of $M_S$.
  
  We first check efficiency.
  Let $S\in 2^N\backslash \{\emptyset\}$ and $M_S$ be the unique stable matching of $(G[S],\prec_S)$.
  Since every free rider has the lowest rank in any preference, $M_S$ is also a maximum matching of $G[S]$.
  Then we have
  \begin{equation*}
    \sum_{i\in S}a_{S,i}=\lvert M_S\rvert=\nu(G[S])=\tau(G[S])=\gamma(S).
  \end{equation*}
  Hence the efficiency follows.

  We now prove monotonicity.
  Let $S,T\in 2^N\backslash \{\emptyset\}$ with $S\subseteq T$.
  Let $M_S$ and $M_T$ be the unique stable matching of $(G[S],\prec_S)$ and $(G[T],\prec_T)$, respectively.
  Let $i\in S$.
  If $i\in M_T$, then we have $i\in M_S$, implying that
  \begin{equation*}
    a_{S,i}=a_{T,i}=1.
  \end{equation*}
 Otherwise, we have
  \begin{equation*}
    a_{S,i}\geq a_{T,i}=0.
  \end{equation*}
\iffalse  
  (If $i$ has the highest rank in any preference of $\prec_T$, then $i$ also has the highest rank in any preference of $\prec_S$.
  It follows that $i\in M_S\cap M_T$, implying that 
  \begin{equation}
    1=a_{S,i}=a_{T,i}.
  \end{equation}
  If $i$ does not have the highest rank in a preference of $\prec_T$, then $i\not\in M_T$ follows, implying that
  \begin{equation}
    a_{S,i}\geq a_{T,i}=0.
  \end{equation})
\fi 
In either case we have $a_{S,i}\geq a_{T,i}$.
Hence the monotonicity follows.
\end{proof}

Theorem \ref{thm:PMAS2SM} reveals the principle behind integral PMAS-es for vertex cover games:
the new player always takes the full cost of the vertex covering it in a minimum vertex cover.
More specifically, consider the following coalition growing process.
The coalition starts with an empty set.
A new player joins the coalition by deciding whether or not to connect to a vertex in the current coalition.
%If the new player decides to connect to a vertex in the current coalition, only a vertex in a minimum vertex cover or a pendant vertex in a star for the current coalition can be connected to.
But every new player is only allowed to connect to a vertex of an existing minimum vertex cover or a pendant vertex of an existing star in the current coalition.
However, the new player always has to pay for the full cost of the vertex covering it, no matter whether or not the vertex belongs to a minimum vertex cover for the coalition before the new player joins.
Theorem \ref{thm:PMAS2SM} also suggests that for vertex cover games,
there is a one-to-one correspondence between integral PMAS-es and preference systems such that every free rider has the lowest rank.
Since such preference system defined on a $(K_3,C_4,P_5)$-free graph has a unique stable matching, we have the following corollary.

\begin{corollary}
Integral PMAS-es for vertex cover games can be enumerated by employing Gale-Shapely algorithm.
\end{corollary}

\section{Concluding remarks}
\label{sec:ending}
This paper investigates PMAS-es for vertex cover games.
We show that the population monotonicity can be determined in polynomial time and a PMAS, if exists, can be constructed accordingly.
We also show that integral PMAS-es can be characterized with stable matchings and be enumerated by employing Gale-Shapley algorithm.
However, neither computing a PMAS nor determining whether a given vector is a PMAS can be done efficiently, as both problems have exponential size.
Nevertheless, integral allocations for each coalition in a PMAS can be computed efficiently by Gale-Shapley algorithm.

\section*{Acknowledgments}

We would like to thank Xin Chen for the valuable help at the early stage of this paper.
We would also like to thank Bo Li for bringing up the notion of free riders and thank Dachuan Xu and Donglei Du for their helpful discussion, which greatly improved the presentation of this paper.

\bibliographystyle{abbrv}
\bibliography{reference}
\end{document}